\documentclass[twocolumn]{autart}    
\usepackage{graphicx}
\usepackage[nodisplayskipstretch]{setspace}
\usepackage{amsfonts,latexsym,
amssymb,
amsmath,url,stackengine}
\usepackage{booktabs}
\usepackage{amsmath}
\usepackage{picins}
\usepackage{color}
\usepackage{empheq}
\usepackage{float}

\usepackage{tikz}
\usetikzlibrary{arrows.meta}

\tikzset{
  bnd/.style={line width=0.5pt},
  bndin/.style={line width=1.4pt, blue!70!black},
  bndout/.style={line width=1.4pt, red!70!black},
  charline/.style={line width=0.3pt, dashed, dash pattern=on 2pt off 1.5pt},
  arr/.style={line width=0.5pt, -{Stealth[length=2.2mm]}},
  normarr/.style={line width=0.5pt, -{Stealth[length=2mm]}, black!70},
  lbl/.style={font=\small\itshape},
  sharpmark/.style={draw, circle, inner sep=0pt, minimum size=4pt, line width=0.5pt, fill=white}
}

\catcode`\@=11
\def\downparenfill{$\m@th\braceld\leaders\vrule\hfill\bracerd$}
\def\downparenfill{$\m@th\braceld\leaders\vrule\hfill\bracerd$}
\def\overparen#1{\mskip 2mu\mathop{\vbox{\ialign{##\crcr\crcr \noalign{\kern0.4ex}
\downparenfill\crcr\noalign{\kern0.4ex\nointerlineskip}
$\hfil\displaystyle{#1}\hfil$\crcr}}}\limits\mskip 2mu} 
\catcode`\@=12

\usepackage{rgsMacros}
\usepackage{subcaption}
\usepackage{bm}

\usepackage{enumitem}
\usepackage{balance}
\usepackage{comment}

\definecolor{olivegreen}{rgb}{0,0.5,0}

\newtheorem{thmm}{Theorem}
\newtheorem{propo}{Proposition}
\newtheorem{lemm}{Lemma}
\newtheorem{exe}{Example}
\newtheorem{corol}{Corollary}
\newtheorem{ass}{Assumption}
\newtheorem{prope}{Property}
\newtheorem{defin}{Definition}
\newtheorem{remm}{Remark}

\newenvironment{remark}{\begin{remm} \rm}{\hfill \end{remm}}

\newenvironment{assumption}{\begin{ass}}{\end{ass}}

\newenvironment{theorem}{\begin{thmm}}{\hfill $\square$ \end{thmm}}

\newenvironment{definition}{\begin{defin}}{\hfill \end{defin}}

\newenvironment{proof}{{\it Proof. }}{\hfill $\blacksquare$ }

\newtheorem{dwellt}{Condition}

\newtheorem{feedback}{Feedback}

\newtheorem{adapt_alg}{Algorithm}

\newif\ifitsdraft

\makeatletter
\def\@citex[#1]#2{%
  \let\@citea\@empty
  \@cite{\@for\@citeb:=#2\do
    {\@citea\def\@citea{,\penalty\@m\hskip.1em}%
     \edef\@citeb{\expandafter\@firstofone\@citeb\@empty}%
     \if@filesw\immediate\write\@auxout{\string\citation{\@citeb}}\fi
     \@ifundefined{b@\@citeb}{\mbox{\reset@font\bfseries ?}%
        \G@refundefinedtrue
        \@latex@warning
          {Citation `\@citeb' on page \thepage \space undefined}}%
        {\hbox{\csname b@\@citeb\endcsname}}}}{#1}}
\makeatother
 
\begin{document}

\setlength{\abovedisplayskip}{6pt}
\setlength{\belowdisplayskip}{6pt}
\setlength{\abovedisplayshortskip}{3pt}
\setlength{\belowdisplayshortskip}{3pt}

\begin{frontmatter}
\title{Backstepping Control of First-Order Hyperbolic Equations in Arbitrary Dimensions with Non-Trapping Characteristics} 

\author{Mohamed Camil Belhadjoudja}\ead{m2camilb@uwaterloo.ca}

\address{Department of Applied Mathematics, University of Waterloo, 200 University Avenue West, Waterloo, ON, Canada, N2L 3G1}


\begin{abstract}
This paper presents a backstepping approach for the boundary control of first-order hyperbolic equations with spatially varying coefficients posed on domains of arbitrary dimension. The method is based on a change of variables induced by the characteristic flow of the time-invariant transport operator, transforming the original multidimensional system into a continuum of decoupled one-dimensional hyperbolic equations evolving along individual characteristic curves. A backstepping controller is then designed for each equation in the decomposition, and the resulting control laws are reassembled in the original coordinates to achieve finite-time stabilization of the full system. The framework relies on the existence of characteristic curves foliating the spatial domain, with uniformly bounded transit times (non-trapping).
\end{abstract}

\end{frontmatter}

\section{Introduction} \label{Introduction}

First-order hyperbolic partial differential equations (PDE)s model the transport of quantities in a broad range of applications, including traffic networks, fluid flows, heat transport, reactive processes, and population dynamics. While one-dimensional reductions may provide useful approximations in some settings, many systems of practical interest are intrinsically higher-dimensional.

Eighteen years ago, \cite{krstic_H} introduced a backstepping method for the
boundary control of general one-dimensional first-order hyperbolic equations.
This seminal contribution unleashed a large body of work and a series of
celebrated successes in the control of first-order hyperbolic systems (see the non-exhaustive list of references \cite{traffic_book,coron_paper,first_c,meglio,hu,auriol,auriol2,auriol3,irscheid,ghoussein,ghoussein2,guan,aamo,ensemble0,ensemble1,ensemble2,federico,bernard,volt_H,back_rob,coron_system} and the survey \cite{krstic_review}). Yet, despite
these advances, the backstepping theory has remained confined to the one-dimensional
setting.

To the best of our knowledge, extensions of the backstepping method to the
multidimensional case have focused exclusively on parabolic equations
\cite{liu,meurer,vaz_key1,vaz_key2,vaz_2,vaz_3,vaz_4,vaz_5,vaz_6,vaz_7,vaz_8,camil}. With the exception of \cite{liu,meurer,vaz_key1,vaz_key2,vaz_3,camil}, the aforementioned works exploit
the fact that, for very specific geometries satisfying symmetry assumptions or translational invariance,
and under relatively restrictive structural conditions on the PDE coefficients,
the original multidimensional parabolic equation can be transformed into a
sequence of decoupled one-dimensional parabolic equations. A backstepping
controller is then designed for each equation, and the resulting controllers
are assembled into a single controller for the original system. While such a
decomposition appears to be available only in highly particular cases for
parabolic equations, the present paper introduces a change of variables that
achieves an analogous decomposition for a \textit{fairly general} class of
multidimensional first-order hyperbolic equations, without any symmetry
assumption whatsoever, yielding a continuum of decoupled one-dimensional
first-order hyperbolic equations to which the approach of \cite{krstic_H}
directly applies.

Our change of variables is rooted in the method of characteristics \cite[Sections 2.1 and 3.2]{evans},
which is classically employed to solve multidimensional transport equations by
reducing them to ordinary differential equations (ODE)s whose solutions trace
the so-called characteristic curves. For
the evolution transport equation $\partial v/\partial t = a(x)\cdot \nabla v$,
these characteristic curves are defined in the space-time domain. Our
construction departs from this classical picture by discarding the time
variable in the very definition of characteristics: we consider the
characteristic curves of the \emph{time-invariant} transport equation
$a(x)\cdot \nabla v = 0$. Each such curve is parametrized by two variables
$(\sigma, \rho)$, where $\rho$ identifies the curve and $\sigma$ denotes the position along it. We prove that a general multidimensional
first-order hyperbolic equation of the form
$\partial v/\partial t = a(x)\cdot \nabla v + (\text{destabilizing terms})$
can be mapped, through the change of variables $x\leftrightarrow (\sigma,\rho)$, into a continuum of
one-dimensional first-order hyperbolic equations evolving along these
characteristic curves and expressed in the variables $(\sigma, t)$, with the
family indexed by $\rho$. The controller of \cite{krstic_H} is then applied
to each PDE in the resulting decomposition. The sole assumption required to
conclude finite-time stability of the original multidimensional system from
that of the decomposed PDEs is that the transit time of characteristics ---
the range of $\sigma$ --- be uniformly bounded across all characteristic
curves. This condition could fail, for example, when the velocity field generates limit cycles, vanishes in the interior of
the spatial domain, or becomes arbitrarily small along certain trajectories,
causing characteristics to stagnate.

It is worth emphasizing that the control of a continuum of hyperbolic
equations --- also referred to as an ensemble of hyperbolic equations --- is
not new, and this is emphatically not the contribution of the present paper.
Our contribution is to establish that the boundary control of a multidimensional
hyperbolic equation reduces, through our
change of variables, to the control of an ensemble of one-dimensional hyperbolic
equations. To the best of our knowledge, this relationship is identified here
for the first time. Prior works on the control of ensembles of hyperbolic
equations have been driven by entirely different motivations and have led to
structurally different ensembles. For instance, \cite{ensemble0,ensemble1,ensemble3} design backstepping
controllers for ensembles of equations and subsequently discretize the
resulting kernel to obtain a controller for a large-scale finite coupled
hyperbolic system; there, the ensemble formulation is a computational tool for
managing complexity in large-scale systems. By contrast, \cite{ensemble2} considers ensembles arising naturally from modeling: physical systems in which agents--vehicles, drivers, particles--are continuously parametrized, and whose collective dynamics are more faithfully captured by a continuum model. Beyond addressing, for the first time since the work of
\cite{krstic_H}, the problem of backstepping control for multidimensional hyperbolic
equations, the present paper thus offers a new and fundamental motivation for
the emerging theory of control of ensembles (and ensembles of ensembles) of hyperbolic equations.

\textit{Notations.} Given a function $v:\mathbb{R}^n\times [0,+\infty)\to \mathbb{R}$, $(x,t)\mapsto v(x,t)$, with $x=(x_1,x_2,...,x_n)\in \mathbb{R}^n$, we let 
$$\nabla v = \left(\frac{\partial v}{\partial x_1},\frac{\partial v}{\partial x_2}, ..., \frac{\partial v}{\partial x_n}\right),$$
where $\partial v/\partial x_i$ denotes the partial derivative of $v$ with respect to $x_i$. Similarly, we denote by $\partial v/\partial t$ the partial derivative of $v$ with respect to $t$. We denote by $v(\cdot,t)$ the function $x\mapsto v(x,t)$. Given two vectors $x,y\in \mathbb{R}^n$ with $x=(x_1,x_2,...,x_n)$ and $y=(y_1,y_2,...,y_n)$, we let 
$x\cdot y = x_1 y_1 + x_2 y_2 + ... + x_n y_n$, and $|x| = \sqrt{x_1^2+x_2^2+...+x_n^2}$.

\section{System description}\label{system_description}
This paper extends some of the results of \cite{krstic_H} to arbitrary spatial dimensions $n \geq 1$. The equation considered in the aforementioned reference is
\begin{align}
\frac{\partial v}{\partial t} =&~ \frac{\partial v}{\partial x} + \lambda(x) v + g(x)v(0) \nonumber \\
&~+ \int_0^x f(x,y)v(y)\,dy \qquad x \in (0,1), \ t>0, \label{p_eq1}
\end{align}
subject to the boundary condition
\begin{align}
v(1,t) = U(v(\cdot,t)) \qquad t\geq 0. \label{p_eq2}
\end{align}
Here, the function $v : [0,1] \times [0,+\infty) \to \mathbb{R}$ denotes the state, the coefficients $\lambda$, $g$, and $f$ are assumed to be continuous, and $U(v(\cdot,t)) \in \mathbb{R}$ is the control variable.

When $g$ and $f$ are positive and sufficiently large, the equilibrium $\{v=0\}$ for \eqref{p_eq1}--\eqref{p_eq2} is unstable under zero input $U = 0$. One of the main contributions of \cite{krstic_H} is the design of a finite-time stabilizing control law $U$ using a backstepping approach.

In this section, we formulate a higher-dimensional analogue of the control system \eqref{p_eq1}--\eqref{p_eq2}. Since the underlying equation is hyperbolic, particular care is required to ensure that the resulting system is well-posed. For example, the location of the control must reflect both the geometry of the domain and the direction of the \textit{velocity field}, i.e., \textit{minus} the coefficient of the first-order spatial derivative term in \eqref{p_eq1}. In one dimension, the velocity reduces to a scalar, so the control location (either at $x=0$ or $x=1$) is determined solely by its sign.

\subsection{Geometry, velocity, and characteristics}

Let $\Omega \in \mathbb{R}^n$, with $n\in \mathbb{N}^*$, be a bounded open and connected set, that satisfies the following assumption.

\begin{assumption}\label{ass1}
$\Omega$ is of class $\mathcal{C}^1$ \cite[Page 626]{evans}.
\hfill $\bullet$
\end{assumption}

Our system will be defined on $\Omega$.

Since $\Omega$ is of class $\mathcal{C}^1$, it admits a unique unit outward normal $\nu : \partial\Omega \to \mathbb{R}^n$, where $\partial \Omega$ is the boundary of $\Omega$. 

A first-order hyperbolic equation on $\Omega$ describes (among other phenomena) the transport of a quantity, such as a concentration, a density, or a temperature, through the domain $\Omega$ with a given velocity field $-a$, where 
\begin{align*}
&a:\mathbb{R}^n\to \mathbb{R}^n, \\ 
&x=(x_1,x_2,...,x_n) \mapsto a(x) = (a_1(x),a_2(x),...,a_n(x)).
\end{align*}
The following regularity assumption is imposed on $a$.
\begin{assumption}\label{ass2}
$a$ is lipschitz.
\hfill $\bullet$
\end{assumption}
The sign of $\nu \cdot a$ determines whether the transported quantity enters the domain, leaves, or travels 
tangentially along $\partial\Omega$: if $\nu(x)\cdot a(x)<0$, the velocity field points outward at $x$ 
and the quantity is leaving $\Omega$ (recall that the velocity field is by convention $-a$, not $a$); if $\nu(x)\cdot a(x)>0$, it points inward and 
the quantity is entering $\Omega$; and if $\nu(x)\cdot a(x)=0$, the velocity is 
tangent to $\partial\Omega$ at $x$ and the quantity does not cross the boundary there. 
Accordingly, we define the following subsets of $\partial \Omega$:
\begin{itemize}
    \item \textit{Inflow boundary} 
    $$\Gamma^+ = \{x\in \partial \Omega \ : \ \nu(x)\cdot a(x) >0\},$$
    \item \textit{Ouflow boundary} 
    $$\Gamma^- = \{x\in \partial \Omega \ : \ \nu(x) \cdot a(x)<0\},$$
    \item \textit{Tangential boundary} 
    $$\Gamma^o =\{x\in \partial \Omega \ : \ \nu(x)\cdot a(x)=0\}.$$
\end{itemize}
Note that $\Gamma^+$, $\Gamma^-$, and $\Gamma^o$ are mutually disjoint by construction and together cover the entire boundary:
$$ \partial \Omega = \Gamma^+ \cup \Gamma^- \cup \Gamma^o.$$

We now define the characteristics associated to $a$. For each $x \in \bar{\Omega}$, where $\bar{\Omega}$ denotes the closure of $\Omega$ with respect to the Euclidean norm $|\cdot |$, the characteristic curve 
passing through $x$ is the unique complete classical solution 
$s \in \mathbb{R} \mapsto \phi(s;x) \in \mathbb{R}^n$ to the ODE
\begin{align}
\frac{d \phi}{ds} = a(\phi), \quad \phi(0;x)=x, \label{char_eq}
\end{align}
whose existence and uniqueness follow from the Cauchy--Lipschitz theorem, applicable 
here by Assumption~\ref{ass2}. Its \textit{entry time} $\tau^+(x)$ and 
\textit{exit time} $\tau^-(x)$ are defined as
\begin{align}
\tau^+(x) &= \inf\{s\geq 0 \ : \ \phi(s;x)\in \Gamma^+\}, \label{4}\\
\tau^-(x) &= \sup\{s\leq 0 \ : \ \phi(s;x)\in \Gamma^-\}, \label{5}
\end{align}
whenever these quantities exist. 

\begin{figure}[H]
    \centering
    \includegraphics[width=\linewidth]{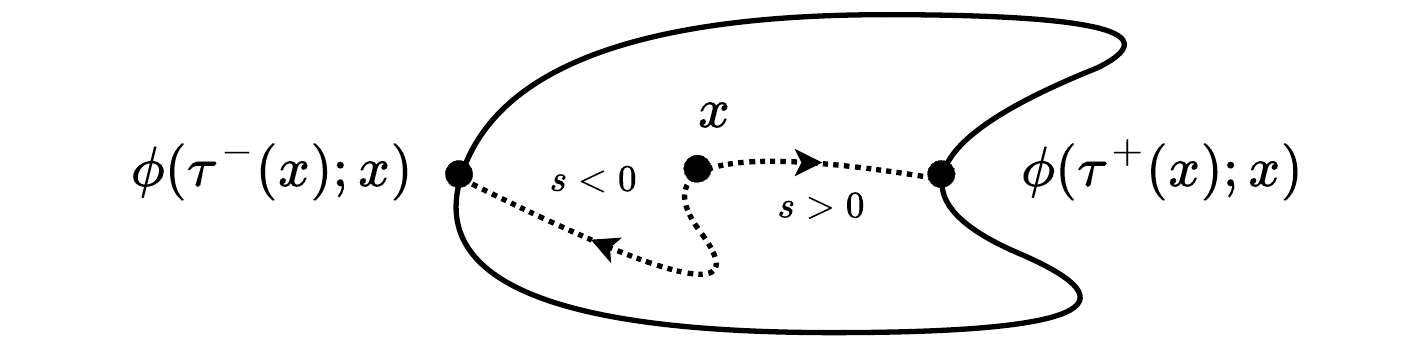}
\end{figure}

We impose the following key assumption on the characteristic curves. A similar condition appears in the theory of geometric optics, where the characteristic curves are the light rays (the solutions to the Eikonal equation), and where it is referred to as the \textit{non-trapping condition}: every ray must exit the medium in finite time \cite{optics}.
\begin{assumption}[Non-trapping characteristics]\label{ass4}
For all $x\in \bar{\Omega}\setminus \Gamma^o$, $\tau^+(x)$ and $\tau^-(x)$ exist and are finite. Moreover, 
\begin{align*}
T_{\max} = \sup_{x\in \bar{\Omega}\setminus \Gamma^o} \{\tau^+(x) - \tau^-(x)\} < +\infty.
\end{align*}
\hfill $\bullet$
\end{assumption}

\subsection{First-order hyperbolic equations on $\Omega$}
As an extension of \eqref{p_eq1}-\eqref{p_eq2} to higher dimensions, we consider on $\Omega$ the first-order hyperbolic equation 
\begin{align}
\frac{\partial v}{\partial t} =&~ a(x)\cdot \nabla v + \lambda(x)v + g(x)v(\rho(x)) \nonumber \\
&~+ \int_{D(x)} f(x,y)v(y) \ d\ell (y), \label{eq_H} 
\end{align}
where $v: \bar{\Omega}\times [0,+\infty) \to \mathbb{R}$ is the state, $\rho(x)$ is the \textit{recirculation point} 
\begin{align}
\rho(x) = \phi(\tau^-(x);x),
\end{align}
$D(x)$ is the \textit{downstream curve} 
\begin{align}
D(x) = \{\phi(s;x) \ : \ \tau^-(x)\leq s \leq 0\},
\end{align}
and $d\ell(y)$ is the arc length element along $D(x)$. Note that 
\begin{align}
&\int_{D(x)} f(x,y)v(y) \ d\ell (y) \nonumber \\
&= \int_{\tau^-(x)}^0f(x,\phi(s;x))v(\phi(s;x)) |a(\phi(s;x))|\ ds.
\end{align}
The coefficients $\lambda, g : \bar{\Omega} \to \mathbb{R}$ and $f:\bar{\Omega}\times \bar{\Omega} \to \mathbb{R}$ satisfy the following assumption. 
\begin{assumption}\label{ass5}
For each $\rho \in \Gamma^-$, the maps
\begin{align*}
\sigma \in [0,\tau^+(\rho)] &\mapsto \lambda(\phi(\sigma;\rho)), \\
\sigma \in [0,\tau^+(\rho)] &\mapsto g(\phi(\sigma;\rho)), \\
(\sigma, \sigma') \in [0,\tau^+(\rho)]^2 &\mapsto f(\phi(\sigma;\rho), \phi(\sigma';\rho)),
\end{align*}
are continuous.
\hfill $\bullet$
\end{assumption}
\begin{remark}
Note that since $\phi(\cdot;\rho)$ is continuous for each $\rho\in \Gamma^-$, then Assumption \ref{ass5} holds trivially if, e.g., $\lambda$, $g$, and $f$ are continuous on their domains of definition. 
\hfill $\bullet$
\end{remark}
Finally, we impose the boundary condition 
\begin{align}
v(z,t) = U(v(\cdot,t),z) \qquad z\in \Gamma^+, \ t\geq 0, \label{bc}
\end{align}
where $U(v(\cdot,t),t)\in \mathbb{R}$ is a control variable to be designed to stabilize $\{v=0\}$ in finite time, in a sense to be specified later. In the sequel, to simplify the notations, we may write $U(z,t)$ instead of $U(v(\cdot,t),z)$.

\begin{remark}
Instead of \eqref{eq_H}, one could consider an equation in which $\rho$ is a general map, and the integration domain $D(x)$ is an arbitrary curve. However, the approach developed in this paper does not apply directly to that setting, to the best of our knowledge.
\hfill $\bullet$
\end{remark}

\section{Concept of solutions and reduction to a continuum of one-dimensional first-order hyperbolic 
equations}

As mentioned above, the objective is to achieve finite-time stabilization of 
$\{v = 0\}$. To state this objective precisely, we first need to define carefully 
the concept of solutions under consideration. A natural first candidate is the 
notion of classical solution.

\begin{definition}[Classical solution]\label{def_classical}
Let $v_o \in C(\bar{\Omega}\setminus\Gamma^o)$. Given a feedback law $U$, a \emph{classical solution} to 
\eqref{eq_H}--\eqref{bc} starting from $v_o$ is any function 
$v \in C((\bar{\Omega}\setminus\Gamma^o) \times [0,+\infty)) \cap 
C^1(\Omega \times (0,+\infty))$ that satisfies 
\eqref{eq_H} for all $(x,t) \in \Omega \times (0,+\infty)$, 
\eqref{bc} for all $(z,t) \in \Gamma^+ \times [0,+\infty)$, and 
$v(x,0) = v_o(x)$ for all $x \in \bar{\Omega}\setminus\Gamma^o$.
\hfill $\bullet$
\end{definition}

Let us suppose for now that we are given a classical solution $v$ to 
\eqref{eq_H}--\eqref{bc}. Roughly speaking, we will show that the restriction of $v$ to each characteristic curve solves a one-dimensional equation of the form \eqref{p_eq1}--\eqref{p_eq2}. The key idea is to first \textit{straighten out} the characteristics of
\eqref{eq_H}: rather than working directly on $\Omega$ with the spatial coordinates
$x = (x_1,x_2,\ldots,x_n)$, we introduce a change of variables that replaces $x\in \Omega$ by a
pair $(\sigma,\rho)$, where $\rho \in \Gamma^-$ identifies
\textit{which} characteristic curve passes through $x$, and $\sigma \geq 0$ records
\textit{how long} the flow has been traveling along that curve from $\rho$ to $x$. In these new coordinates, the transport term $a(x) \cdot \nabla v$ in
\eqref{eq_H} reduces to a simple partial derivative with respect to $\sigma$, and the
full $n$-dimensional system \eqref{eq_H}--\eqref{bc} decouples into a continuum of
one-dimensional systems of precisely the same form as \eqref{p_eq1}--\eqref{p_eq2},
each one evolving along a characteristic curve parameterized by $\rho$.

We now make this construction precise. For each $x \in \bar{\Omega}\setminus \Gamma^o$, we define the
\textit{time to exit}
\begin{align}
    \sigma(x) = -\tau^-(x) \geq 0,
\end{align}
that is, $\sigma(x)$ is the time the characteristic flow takes to travel from $x$ to $\rho(x) = \phi(\tau^-(x);x) \in \Gamma^-$. In particular, $\sigma(x) = 0$ if and only if
$x \in \Gamma^-$, and
more generally $\sigma(x) \in [0, T(x)]$, where $T(x)$ is the transit time of the characteristic passing through $x$,
defined as $T(x) = \tau^+(x) - \tau^-(x) > 0$.

Note that, since by Assumption~\ref{ass2} the ODE \eqref{char_eq} admits unique
solutions, characteristic curves cannot intersect: if two solutions of \eqref{char_eq}
coincide at any single time, they must coincide for all times. Therefore
each point $x \in \bar{\Omega}\setminus \Gamma^o$ belongs to exactly one characteristic, and the
transit time is constant along that characteristic: $T(\phi(s;\rho)) = T(\rho) = \tau^+(\rho)$ for all $s\in[0,T(\rho)]$ and all $\rho\in\Gamma^-$.

Next, we introduce the \textit{transformed domain}
\begin{align}
    \hat{\Omega} = \left\{(\sigma,\rho) \in \mathbb{R} \times \Gamma^- \ : \
    0 \leq \sigma \leq T(\rho) \right\},
\end{align}
with $T(\rho)$ the transit time of the characteristic passing through $\rho$, and the map
\begin{align}
    \Psi : \bar{\Omega}\setminus \Gamma^o \to \hat{\Omega}, \qquad
    x \mapsto \Psi(x) = \left(\sigma(x),\, \rho(x)\right).
\end{align}
The map $\Psi$ is bijective. Indeed, given any $(\sigma, \rho) \in \hat{\Omega}$,
the point $x = \phi(\sigma;\rho)$ satisfies $\rho(x) = \rho$ and $\sigma(x) = \sigma$, so $\Psi$ is surjective. Injectivity follows
from the non-intersection of characteristics: if $\Psi(x_1) = \Psi(x_2)$, then
$x_1$ and $x_2$ lie on the same characteristic and reach the exit point $\rho(x_1)=\rho(x_2)$
at the same time, hence $x_1 = x_2$. The inverse of $\Psi$ is therefore $\Psi^{-1}(\sigma,\rho) = \phi(\sigma;\rho)$.

We now use $\Psi$ to rewrite \eqref{eq_H}--\eqref{bc} in the coordinates
$(\sigma,\rho) \in \hat{\Omega}$. To this end, we introduce the transformed state
$u : \hat{\Omega} \times [0,+\infty) \to \mathbb{R}$ defined by
\begin{align}
    u(\sigma, \rho, t) = v\!\left(\Psi^{-1}(\sigma,\rho),\, t\right)
    = v\!\left(\phi(\sigma;\rho),\, t\right).
    \label{u-def}
\end{align}
In particular, $u(0,\rho,t) = v(\rho,t)$ is
the value of $v$ at the exit point $\rho$, and $u(T(\rho),\rho,t) =
v(\phi(T(\rho);\rho),t)$ is the value of $v$ at the corresponding entry point on
$\Gamma^+$.

\begin{figure}[H]
    \centering
    \includegraphics[width=0.9\linewidth]{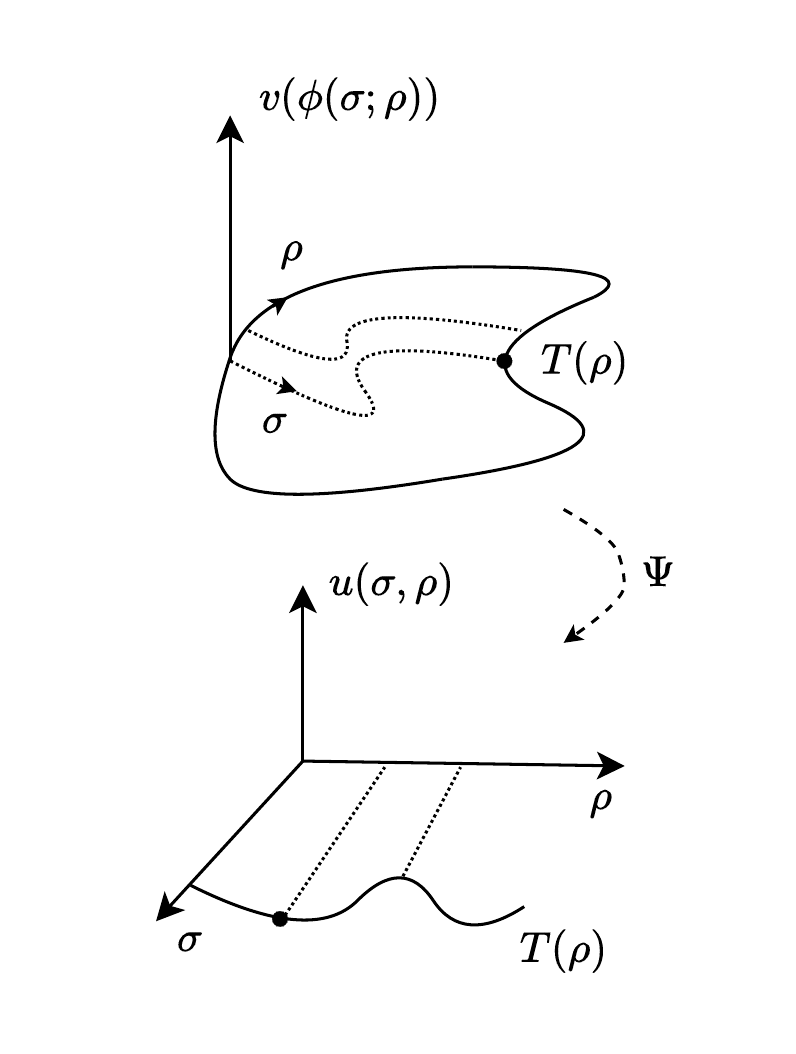}
\end{figure}

The following main result shows that $u$ satisfies a continuum of equations of exactly the
same form as \eqref{p_eq1}--\eqref{p_eq2}.

\begin{theorem}\label{thm1}
Let Assumptions \ref{ass1}--\ref{ass4} hold, and let $u$ be defined by \eqref{u-def} with $v$ a classical solution to \eqref{eq_H}--\eqref{bc}. Then, for each $\rho\in \Gamma^{-}$ and for all $\sigma \in (0,T(\rho))$,
\begin{align}
    \frac{\partial u}{\partial t} =&~
    \frac{\partial u}{\partial \sigma}
    + \hat{\lambda}(\sigma,\rho)\, u(\sigma,\rho,t)
    + \hat{g}(\sigma,\rho)\, u(0,\rho,t) \nonumber \\
    &~+ \int_0^\sigma \hat{f}(\sigma,\rho,\sigma')\,
    u(\sigma',\rho,t) \ d\sigma', \label{c_1}
\end{align}
together with the boundary condition
\begin{align}
    u\!\left(T(\rho),\rho,t\right) =
    U\!\left(\phi(T(\rho);\rho),\, t\right) \qquad t \geq 0,
    \label{c_2}
\end{align}
where the transformed coefficients are given by
\begin{align*}
    \hat{\lambda}(\sigma,\rho) &=
    \lambda\!\left(\Psi^{-1}(\sigma,\rho)\right), \\
    \hat{g}(\sigma,\rho) &=
    g\!\left(\Psi^{-1}(\sigma,\rho)\right), \\
    \hat{f}(\sigma,\rho,\sigma') &=
    f\!\left(\Psi^{-1}(\sigma,\rho),\,\Psi^{-1}(\sigma',\rho)\right)
    \left|a\!\left(\Psi^{-1}(\sigma',\rho)\right)\right|.
\end{align*}
Conversely, if $u(\cdot,\rho,\cdot)$ satisfies \eqref{c_1}--\eqref{c_2} for every
$\rho\in\Gamma^-$, and if $u\circ \Psi \in C((\bar{\Omega}\setminus\Gamma^o) \times [0,+\infty)) \cap 
C^1(\Omega \times (0,+\infty))$, then $v = u \circ \Psi$ is a classical solution to
\eqref{eq_H}--\eqref{bc} in the sense of Definition \ref{def_classical}.
\end{theorem}
\begin{proof}
We transform each term of \eqref{eq_H} separately after the substitution
$x = \Psi^{-1}(\sigma,\rho) = \phi(\sigma;\rho)$.

\textit{Step 1: time derivative.}
Differentiating both sides of \eqref{u-def} with respect to $t$, and using the
fact that $\Psi^{-1}(\sigma,\rho)$ does not depend on $t$, we obtain
\begin{align}
    \frac{\partial u}{\partial t}(\sigma,\rho,t) =
    \frac{\partial v}{\partial t}\!\left(\Psi^{-1}(\sigma,\rho),t\right).
    \label{step1}
\end{align}
\textit{Step 2: transport term.}
We wish to show that 
$$a(x) \cdot \nabla v(x,t) = \frac{\partial u}{\partial
\sigma}(\sigma,\rho,t)$$
at $x = \Psi^{-1}(\sigma,\rho)$. To this end, we differentiate the function $s \mapsto v(\phi(s;x),t)$, to obtain,
\begin{align}
    \frac{d}{ds}\, v(\phi(s;x),t) &=
    \frac{d\phi(s;x)}{ds} \cdot \nabla v(\phi(s;x),t) \nonumber \\
    &=
    a(\phi(s;x)) \cdot \nabla v(\phi(s;x),t).
    \label{chain}
\end{align}
Evaluating \eqref{chain} at $s = 0$ and using $\phi(0;x) = x$ yields
\begin{align}
    a(x) \cdot \nabla v(x,t) =
    \frac{d}{ds}\bigg|_{s=0} v(\phi(s;x),t).
    \label{transport-identity}
\end{align}
We now set $x = \Psi^{-1}(\sigma,\rho) = \phi(\sigma;\rho)$ in
\eqref{transport-identity}. The right-hand side becomes
\begin{align}
    \frac{d}{ds}\bigg|_{s=0} v\!\left(\phi\!\left(s;\,\phi(\sigma;\rho)\right),t\right).
    \label{rhs-before}
\end{align}
We simplify the argument of $v$ using
\begin{align}
    \phi\!\left(s;\,\phi(\sigma;\rho)\right) = \phi(s+\sigma;\rho).
    \label{flow-comp}
\end{align}
Substituting \eqref{flow-comp} into \eqref{rhs-before}, we obtain
\begin{align*}
\frac{d}{ds}\bigg|_{s=0} v\!\left(\phi\!\left(s;\,\phi(\sigma;\rho)\right),t\right) =   \frac{d}{ds}\bigg|_{s=0} v\!\left(\phi(s+\sigma;\rho),t\right).
\end{align*}
By definition \eqref{u-def}, $v(\phi(s+\sigma;\rho),t) = u(s+\sigma,\rho,t)$.
Substituting and differentiating at $s=0$ gives
\begin{align}
    \frac{d}{ds}\bigg|_{s=0} u(s+\sigma,\rho,t) =
    \frac{\partial u}{\partial \sigma}(\sigma,\rho,t).
    \label{step2}
\end{align}
Combining \eqref{transport-identity}--\eqref{step2}, we conclude
\begin{align}
    a\!\left(\Psi^{-1}(\sigma,\rho)\right) \cdot
    \nabla v\!\left(\Psi^{-1}(\sigma,\rho),t\right) =
    \frac{\partial u}{\partial \sigma}(\sigma,\rho,t).
\end{align}
\noindent\textit{Step 3: reaction term.}
By the definition of $\hat{\lambda}$ and the relation \eqref{u-def}, we have
directly
\begin{align}
    \lambda\!\left(\Psi^{-1}(\sigma,\rho)\right)
    v\!\left(\Psi^{-1}(\sigma,\rho),t\right) =
    \hat{\lambda}(\sigma,\rho)\, u(\sigma,\rho,t).
\end{align}
\noindent\textit{Step 4: recirculation term.}
We need to transform $g(x)\, v(\rho(x),t)$ at $x = \Psi^{-1}(\sigma,\rho)$.
We first identify the recirculation point. Since $\rho \in \Gamma^-$ is the exit point of the characteristic through
$\Psi^{-1}(\sigma,\rho)$, we have
\begin{align}
    \rho\!\left(\Psi^{-1}(\sigma,\rho)\right) = \rho.
    \label{rho-id}
\end{align}
Next, the value of $v$ at $\rho$ is, by \eqref{u-def} evaluated
at $\sigma = 0$,
\begin{align}
    v(\rho, t) = v\!\left(\Psi^{-1}(0,\rho),t\right) = u(0,\rho,t).
    \label{v-at-rho}
\end{align}
Combining \eqref{rho-id}--\eqref{v-at-rho} with the definition of $\hat{g}$,
\begin{align*}
    g\!\left(\Psi^{-1}(\sigma,\rho)\right)
    v\!\left(\rho\!\left(\Psi^{-1}(\sigma,\rho)\right),t\right) =
    \hat{g}(\sigma,\rho)\, u(0,\rho,t).
\end{align*}
\textit{Step 5: nonlocal term.}
We next transform the term
\begin{align}
    \int_{\tau^-(x)}^0 f(x,\phi(s;x))\,
    v(\phi(s;x),t)\, |a(\phi(s;x))| \ ds.
    \label{nonlocal-orig}
\end{align}
We set $x = \Psi^{-1}(\sigma,\rho) = \phi(\sigma;\rho)$, so that $\tau^-(x) =
-\sigma$ and the integration interval is $[-\sigma, 0]$. Using \eqref{flow-comp}, we have 
$$\phi(s;x) = \phi(s;\phi(\sigma;\rho)) =
\phi(s+\sigma;\rho).$$
Then, we perform the change of variables $\sigma' = s + \sigma$,
which maps $s \in [-\sigma,0]$ to $\sigma'\in[0,\sigma]$ and gives
$$\phi(s;x) = \phi(\sigma';\rho) = \Psi^{-1}(\sigma',\rho),\ \ 
v(\phi(s;x),t) = u(\sigma',\rho,t),$$
and $ds = d\sigma'$, to obtain
\begin{align}
    &\int_{\tau^-(x)}^0 f(x,\phi(s;x))\,
    v(\phi(s;x),t)\, |a(\phi(s;x))| \ ds
    \nonumber \\
    &\qquad =
    \int_0^\sigma \hat{f}(\sigma,\rho,\sigma')\,
    u(\sigma',\rho,t) \ d\sigma'.
\end{align}

\noindent\textit{Conclusion.}
Summing the results of Steps 1--5 and using the fact that $v$ satisfies \eqref{eq_H},
we obtain \eqref{c_1}. The boundary condition \eqref{c_2} follows directly from
\eqref{bc} evaluated at $z = \phi(T(\rho);\rho) \in \Gamma^+$.

The converse is immediate by reversing each step and using the assumed regularity of $u\circ \Psi$.
\end{proof}

\begin{remark}
For each fixed $\rho \in \Gamma^-$, \eqref{c_1}--\eqref{c_2} is a one-dimensional 
first-order hyperbolic equation in $(\sigma,t)$ of exactly the same form as 
\eqref{p_eq1}--\eqref{p_eq2}, with spatial variable $\sigma \in [0,T(\rho)]$, 
unit velocity, reaction coefficient $\hat\lambda(\cdot,\rho)$, recirculation 
coefficient $\hat g(\cdot,\rho)$, nonlocal kernel $\hat f(\cdot,\rho,\cdot)$, and 
boundary control $U(\phi(T(\rho);\rho),t)$ applied at $\sigma = T(\rho)$. The 
continuity of the transformed coefficients (see Assumption \ref{ass5}) ensures that \eqref{c_1}--\eqref{c_2} 
falls within the framework of~\cite{krstic_H}.
\hfill $\bullet$
\end{remark}

In light of Theorem~\ref{thm1}, we introduce the following concept of solution 
to \eqref{eq_H}--\eqref{bc}.

\begin{definition}[Characteristic solution]\label{def1}
A characteristic solution to \eqref{eq_H}--\eqref{bc} starting from $v_o : \bar{\Omega}\setminus \Gamma^o \to \mathbb{R}$, where $\sigma \mapsto v_o(\phi(\sigma;\rho))$ belongs to $\mathcal{C}([0,T(\rho)])\cap \mathcal{C}^1(0,T(\rho))$ for each $\rho\in \Gamma^-$, is any function $v : (\bar{\Omega}\setminus\Gamma^o) 
\times [0,+\infty) \to \mathbb{R}$ of the form $v = u \circ \Psi$, 
where $(\sigma,t)\mapsto u(\sigma,\rho,t)$ belongs to $\mathcal{C}([0,T(\rho)]\times [0,+\infty))\cap \mathcal{C}^1((0,T(\rho))\times (0,+\infty))$ for each $\rho \in \Gamma^-$ and verifies \eqref{c_1} for all $(\sigma,t)\in (0,T(\rho))\times (0,+\infty)$, \eqref{c_2} for all $t\geq 0$, and $u(\sigma,\rho,0)= v_o(\phi(\sigma,\rho))$ for all 
$(\sigma,\rho)\in \hat{\Omega}$.
\hfill $\bullet$
\end{definition}

\begin{remark}
By the converse 
result in Theorem~\ref{thm1}, any solution in the sense of Definition~\ref{def1} 
that additionally belongs to $C((\bar{\Omega}\setminus\Gamma^o)\times[0,+\infty)) 
\cap C^1(\Omega\times(0,+\infty))$ is a classical solution. In other words, characteristic solutions that are regular enough must be classical solutions, and classical solutions are necessarily characteristic solutions, which justifies the proposed concept of solutions. It is worth noting that defining solutions to first-order hyperbolic equations through formulas derived from the method of characteristics is not new; see \cite[Definition 3.1]{coron_control}, \cite[Definition 4.1]{chitour_1} and \cite[Definition 3.1]{chitour_2} for related approaches in the context of controllability and stability analysis of systems of one-dimensional hyperbolic equations. The novelty of our formulation lies in the fact that the resulting characterisation is expressed in terms of PDEs that are suitable for PDE-based control techniques, such as backstepping. 
\hfill $\bullet$
\end{remark}

Given the characterization in Theorem \ref{thm1} and Definition \ref{def1}, finite-time stabilization of $\{v=0\}$, which consists of designing the control input $U$ such that $v(x,t)=0$ for all $x\in \bar{\Omega}\setminus \Gamma^o$ and all $t$ beyond a certain finite time, reduces to the finite-time stabilization of each PDE in the continuum representation. It is precisely here that Assumption \ref{ass4} plays a key role, as it ensures the existence of a uniform bound on the finite settling time beyond which the state of each PDE in the continuum representation can be driven to zero.

\section{Backstepping controller}

For each fixed $\rho$, we have a first-order hyperbolic equation on the interval $[0,T(\rho)]$, which is
precisely the class of systems studied in \cite{krstic_H}. The idea is therefore to apply, for each fixed $\rho$, the backstepping controller of \cite{krstic_H}.

First, we normalize the interval $[0,T(\rho)]$ to $[0,1]$. To do so, we introduce the rescaled spatial and time variables
\begin{align}
    \bar{\sigma} = \frac{\sigma}{T(\rho)} \in [0,1], \qquad
    \bar{t} = \frac{t}{T(\rho)} \geq 0,
\end{align}
and define the normalized state
\begin{align}
    \bar{v}(\bar{\sigma}, \rho, \bar{t}) = u\!\left(T(\rho)\bar{\sigma},\,
    \rho,\, T(\rho)\bar{t}\right). \label{barv-def}
\end{align}
To find the equation satisfied by $\bar{v}$, we substitute into \eqref{c_1}. Since
$u(\sigma,\rho,t) = \bar{v}(\sigma/T(\rho), \rho, t/T(\rho))$, we have
\begin{align*}
    \frac{\partial u}{\partial t} = \frac{1}{T(\rho)}
    \frac{\partial \bar{v}}{\partial \bar{t}}, \qquad
    \frac{\partial u}{\partial \sigma} = \frac{1}{T(\rho)}
    \frac{\partial \bar{v}}{\partial \bar{\sigma}}.
\end{align*}
For the nonlocal term, we change variables $\sigma' = T(\rho)\bar{y}$ in the
integral, giving $d\sigma' = T(\rho)\,d\bar{y}$ and
\begin{align*}
    \int_0^\sigma &\hat{f}(\sigma,\rho,\sigma')\,u(\sigma',\rho,t)\,d\sigma' \nonumber \\
    &
    = T(\rho)\int_0^{\bar{\sigma}}
    \hat{f}\!\left(T(\rho)\bar{\sigma},\rho,T(\rho)\bar{y}\right)
    \bar{v}(\bar{y},\rho,\bar{t})\,d\bar{y}.
\end{align*}
Substituting into \eqref{c_1} and multiplying through by $T(\rho)$, we obtain
\begin{align}
    \frac{\partial \bar{v}}{\partial \bar{t}} =&~
    \frac{\partial \bar{v}}{\partial \bar{\sigma}}
    + \bar{\lambda}(\bar{\sigma},\rho)\,\bar{v}(\bar{\sigma},\rho,\bar{t})
    + \bar{g}(\bar{\sigma},\rho)\,\bar{v}(0,\rho,\bar{t}) \nonumber \\
    &~+ \int_0^{\bar{\sigma}}
    \bar{f}(\bar{\sigma},\rho,\bar{y})\,
    \bar{v}(\bar{y},\rho,\bar{t})\,d\bar{y}, \label{c_8_pre}
\end{align}
where the rescaled coefficients are defined by
\begin{align*}
    \bar{\lambda}(\bar{\sigma},\rho) &=
    T(\rho)\,\hat{\lambda}\!\left(T(\rho)\bar{\sigma},\rho\right), \\
    \bar{g}(\bar{\sigma},\rho) &=
    T(\rho)\,\hat{g}\!\left(T(\rho)\bar{\sigma},\rho\right), \\
    \bar{f}(\bar{\sigma},\rho,\bar{y}) &=
    T(\rho)^2\,\hat{f}\!\left(T(\rho)\bar{\sigma},\rho,T(\rho)\bar{y}\right).
\end{align*}
The boundary condition \eqref{c_2} becomes
\begin{align}
    \bar{v}(1,\rho,\bar{t}) =
    U\!\left(\phi(T(\rho);\rho),\,T(\rho)\bar{t}\right). \label{c_9_pre}
\end{align}
Next, we absorb the reaction term into the nonlocal term. Concretely, we define
\begin{align}
    \Lambda(\bar{\sigma},\rho) =
    \int_0^{\bar{\sigma}} \bar{\lambda}(\xi,\rho)\,d\xi, \label{Lambda-def}
\end{align}
and the new state
\begin{align}
    w(\bar{\sigma},\rho,\bar{t}) =
    e^{\Lambda(\bar{\sigma},\rho)}\,\bar{v}(\bar{\sigma},\rho,\bar{t}).
    \label{tilde_v_def_final}
\end{align}
Note that $\Lambda(0,\rho) = 0$, so $w(0,\rho,\bar{t}) =
\bar{v}(0,\rho,\bar{t})$.
To derive the equation for $w$, we compute
\begin{align*}
    \frac{\partial w}{\partial \bar{t}} &=
    e^{\Lambda}\,\frac{\partial \bar{v}}{\partial \bar{t}}, \\
    \frac{\partial w}{\partial \bar{\sigma}} &=
    \bar{\lambda}\,e^{\Lambda}\,\bar{v} +
    e^{\Lambda}\,\frac{\partial \bar{v}}{\partial \bar{\sigma}} =
    \bar{\lambda}\,w +
    e^{\Lambda}\,\frac{\partial \bar{v}}{\partial \bar{\sigma}},
\end{align*}
so that 
$$e^{\Lambda}\,\frac{\partial \bar{v}}{\partial \bar{\sigma}} =
\frac{\partial w}{\partial \bar{\sigma}} - \bar{\lambda}\,w.$$
Multiplying \eqref{c_8_pre} by $e^{\Lambda(\bar{\sigma},\rho)}$ and substituting,
\begin{align*}
    \frac{\partial w}{\partial \bar{t}} =&~
    \left(\frac{\partial w}{\partial \bar{\sigma}} -
    \bar{\lambda}\,w\right) +
    \bar{\lambda}\,w +
    \bar{g}\,e^{\Lambda(\bar{\sigma},\rho)}\,\bar{v}(0,\rho,\bar{t}) \\
    &~+ \int_0^{\bar{\sigma}}
    \bar{f}(\bar{\sigma},\rho,\bar{y})\,
    e^{\Lambda(\bar{\sigma},\rho)}\,\bar{v}(\bar{y},\rho,\bar{t})\,d\bar{y}.
\end{align*}
The $\bar{\lambda}\,w$ terms cancel. Using $\bar{v}(0,\rho,\bar{t}) =
w(0,\rho,\bar{t})$ and
$e^{\Lambda(\bar{\sigma},\rho)}\bar{v}(\bar{y},\rho,\bar{t}) =
e^{\Lambda(\bar{\sigma},\rho)-\Lambda(\bar{y},\rho)}w(\bar{y},\rho,\bar{t})$,
we obtain
\begin{align}
    \frac{\partial w}{\partial \bar{t}} =&~
    \frac{\partial w}{\partial \bar{\sigma}} +
    G(\bar{\sigma},\rho)\,w(0,\rho,\bar{t}) \nonumber \\
    &~+ \int_0^{\bar{\sigma}}
    F(\bar{\sigma},\rho,\bar{y})\,w(\bar{y},\rho,\bar{t})\,d\bar{y},
    \label{c_8}
\end{align}
where
\begin{align*}
    G(\bar{\sigma},\rho) &=
    \bar{g}(\bar{\sigma},\rho)\,e^{\Lambda(\bar{\sigma},\rho)}, \\
    F(\bar{\sigma},\rho,\bar{y}) &=
    \bar{f}(\bar{\sigma},\rho,\bar{y})\,
    e^{\Lambda(\bar{\sigma},\rho)-\Lambda(\bar{y},\rho)}.
\end{align*}
The boundary condition \eqref{c_9_pre} becomes
\begin{align}
    w(1,\rho,\bar{t}) =
    e^{\Lambda(1,\rho)}\,U\!\left(\phi(T(\rho);\rho),\,T(\rho)\bar{t}\right).
    \label{c_9}
\end{align}
Finally, we define the backstepping transformation
\begin{align*}
    \bar{w}(\bar{\sigma},\rho,\bar{t}) = w(\bar{\sigma},\rho,\bar{t}) -
    \int_0^{\bar{\sigma}} k(\bar{\sigma},\bar{y};\rho)\,
    w(\bar{y},\rho,\bar{t})\,d\bar{y},
\end{align*}
where the kernel $k(\cdot,\cdot;\rho)$ (parametrized by $\rho$) is to be determined on the triangle
\begin{align*}
    \mathcal{T} = \left\{(\bar{\sigma},\bar{y}) \in [0,1]^2 \ : \
    0 \leq \bar{y} \leq \bar{\sigma} \leq 1\right\},
\end{align*}
such that $\bar{w}$ solves the \textit{target system}
\begin{align}
    \frac{\partial \bar{w}}{\partial \bar{t}} &=
    \frac{\partial \bar{w}}{\partial \bar{\sigma}},
    \label{target_1} \\
    \bar{w}(1,\rho,\bar{t}) &= 0. \label{target_2}
\end{align}
The solution to the target system verifies
\begin{align}
    \bar{w}(\bar{\sigma},\rho,\bar{t}) = 0 \qquad
    \forall\,\bar{\sigma}\in[0,1],\ \forall\,\bar{t}\geq 1. \label{w-zero}
\end{align}

The kernel equations are the same as in \cite{krstic_H}, namely,
\begin{align}
    \frac{\partial k}{\partial \bar{\sigma}} +
    \frac{\partial k}{\partial \bar{y}} &=
    \int_{\bar{y}}^{\bar{\sigma}} k(\bar{\sigma},\xi;\rho)\,
    F(\xi,\rho,\bar{y})\,d\xi - F,
    \label{kernel_1} \\
    k(\bar{\sigma},0;\rho) &=
    \int_0^{\bar{\sigma}} k(\bar{\sigma},\bar{y};\rho)\,G(\bar{y},\rho)\,d\bar{y}
    - G, \label{kernel_2}
\end{align}
and they admit, for each $\rho \in \Gamma^-$, a unique solution $k(\cdot,\cdot;\rho)
\in C^1(\mathcal{T})$.

To enforce $\bar{w}(1,\rho,\bar{t})=0$, it suffices to let
\begin{align}
    w(1,\rho,\bar{t}) =
    \int_0^1 k(1,\bar{y};\rho)\,w(\bar{y},\rho,\bar{t})\,d\bar{y}.
    \label{control_tilde}
\end{align}

Combining \eqref{control_tilde} with \eqref{c_9}, we obtain
\begin{align*}
e^{\Lambda(1,\rho)}\,U\!\left(\phi(T(\rho);\rho),\,T(\rho)\bar{t}\right) =
    \int_0^1 k(1,\bar{y};\rho)\,w(\bar{y},\rho,\bar{t})\,d\bar{y}.
\end{align*}
Substituting $w(\bar{y},\rho,\bar{t}) = e^{\Lambda(\bar{y},\rho)}
\bar{v}(\bar{y},\rho,\bar{t})$ and dividing by
$e^{\Lambda(1,\rho)}$,
\begin{align}
&U\!\left(\phi(T(\rho);\rho),\,T(\rho)\bar{t}\right) \nonumber \\
&\qquad =
    \int_0^1 k(1,\bar{y};\rho)\,
    e^{\Lambda(\bar{y},\rho)-\Lambda(1,\rho)}\,
    \bar{v}(\bar{y},\rho,\bar{t})\,d\bar{y}. \label{control_barv}
\end{align}
We now revert to the original variables. Setting $t = T(\rho)\bar{t}$ and
$\sigma = T(\rho)\bar{y}$ in \eqref{control_barv}, so that
$d\bar{y} = d\sigma/T(\rho)$ and $\bar{v}(\bar{y},\rho,\bar{t}) =
u(\sigma,\rho,t)$ by \eqref{barv-def}, we obtain
\begin{align}
    U\!\left(\phi(T(\rho);\rho),t\right) =
    \int_0^{T(\rho)} K(\sigma;\rho)\,u(\sigma,\rho,t)\,d\sigma, \label{control_law_u}
\end{align}
where $K:[0,T(\rho)]\times \Gamma^-\to\mathbb{R}$ is defined by
\begin{align*}
    K(\sigma;\rho) := \frac{1}{T(\rho)}\,
    k\!\left(1,\frac{\sigma}{T(\rho)};\rho\right)\,
    e^{\Lambda(\sigma/T(\rho),\rho)-\Lambda(1,\rho)}.
\end{align*}
Finally, we express \eqref{control_law_u} in terms of the original state $v$ on
$\Omega$. Let $z \in \Gamma^+$, and define the \textit{upstream curve} of $z$
as
\begin{align}
    C(z) := \left\{\phi(s;z) \ : \ \tau^-(z) \leq s \leq 0\right\},
\end{align}
i.e., the arc of the characteristic curve from the exit point $\rho(z)$ to $z$,
parameterized by backward time from $z$. As $\sigma$ runs from $0$ to $T(\rho(z))$,
the point $y = \Psi^{-1}(\sigma,\rho(z)) = \phi(\sigma;\rho(z))$ traces $C(z)$,
and the arc length element along $C(z)$ is $d\ell(y) = |a(y)|\,d\sigma$. Using
$u(\sigma,\rho(z),t) = v(\Psi^{-1}(\sigma,\rho(z)),t) = v(y,t)$ and $d\sigma =
d\ell(y)/|a(y)|$, the integral in \eqref{control_law_u} becomes
\begin{align}
    U(z,t) = \int_{C(z)} \mathcal{K}(z,y)\,v(y,t)\,d\ell(y), \label{control_law_x}
\end{align}
where $\mathcal{K}:\{(z,y)\in\Gamma^+\times\bar{\Omega}
: y \in C(z)\}\to\mathbb{R}$ is given by
\begin{align}
    \mathcal{K}(z,y) = \frac{K(\sigma(y);\rho(z))}{|a(y)|}. \label{gain_K}
\end{align}

Finally, let the initial condition $v_o : \bar{\Omega}\setminus \Gamma^o \to \mathbb{R}$, where $\sigma \mapsto v_o(\phi(\sigma;\rho))$ belongs to $\mathcal{C}([0,T(\rho)])\cap \mathcal{C}^1(0,T(\rho))$ for each $\rho\in \Gamma^-$. By the invertibility of the backstepping transformation for each $\rho$ and the $\mathcal{C}^1$ regularity of the corresponding inverse backstepping kernel \cite{krstic_H}, each equation in the continuum admits a unique classical solution starting from $v_o(\phi(\cdot;\rho))$ provided the compatibility condition $v_o(\phi(\sigma;\rho)) = U(\phi(\sigma;\rho),0)$ holds, and therefore the system \eqref{eq_H}--\eqref{bc} admits a unique characteristic solution starting from $v_o$.  

We can now state the following result.

\begin{theorem}\label{thm_main}
Consider the system \eqref{eq_H}-\eqref{bc} and let Assumptions \ref{ass1}--\ref{ass5} hold. Furthermore, let the control input be given by \eqref{control_law_x}. Then, for any initial condition $v_o : \bar{\Omega}\setminus \Gamma^o \to \mathbb{R}$, where $\sigma \mapsto v_o(\phi(\sigma;\rho))$ belongs to $\mathcal{C}([0,T(\rho)])\cap \mathcal{C}^1(0,T(\rho))$ for each $\rho\in \Gamma^-$, and verifying the compatibility condition $v_o(z)=U(v_o,z)$ on $\Gamma^+$, the system \eqref{eq_H}-\eqref{bc} admits a unique characteristic solution starting from $v_o$ that satisfies
\begin{align*}
    v(x,t) = 0 \qquad \forall\,x \in \bar{\Omega} \setminus \Gamma^o,\ \forall\,t \geq T_{\max}.
\end{align*}
\end{theorem}

\section{Example}

Let $\Omega$ be the unit disk 
\begin{align*}
\Omega = \{x\in \mathbb{R}^2 \ : \ |x|<1\},
\end{align*}
whose unit outward normal is 
\begin{align*}
\nu (x) = x \qquad \forall x\in \partial \Omega.
\end{align*}
Furthermore, let 
\begin{align*}
a = (1,1).
\end{align*}
Then, by expanding the inner product $a\cdot \nu$, we have 
\begin{align*}
\Gamma^+ &= \{(x_1,x_2)\in \partial \Omega \ : \ x_1>-x_2\}, \\
\Gamma^- &= \{(x_1,x_2)\in \partial \Omega \ : \ x_1<-x_2\}, \\
\Gamma^o &= \{(x_1,x_2)\in \partial \Omega \ : \ x_1=-x_2\}.
\end{align*}
Note that Assumptions \ref{ass1}--\ref{ass2} hold.

The characteristic curve passing through $x\in \bar{\Omega}$ is the unique complete solution to 
\begin{align*}
\frac{d \phi}{ds} = (1,1), \qquad \phi(0;x) = x.
\end{align*}
That is, 
\begin{align*}
\phi(s;x) = (s,s)+x \qquad \forall s\in \mathbb{R}.
\end{align*}
\begin{figure}[H]
    \centering
    \includegraphics[width=\linewidth]{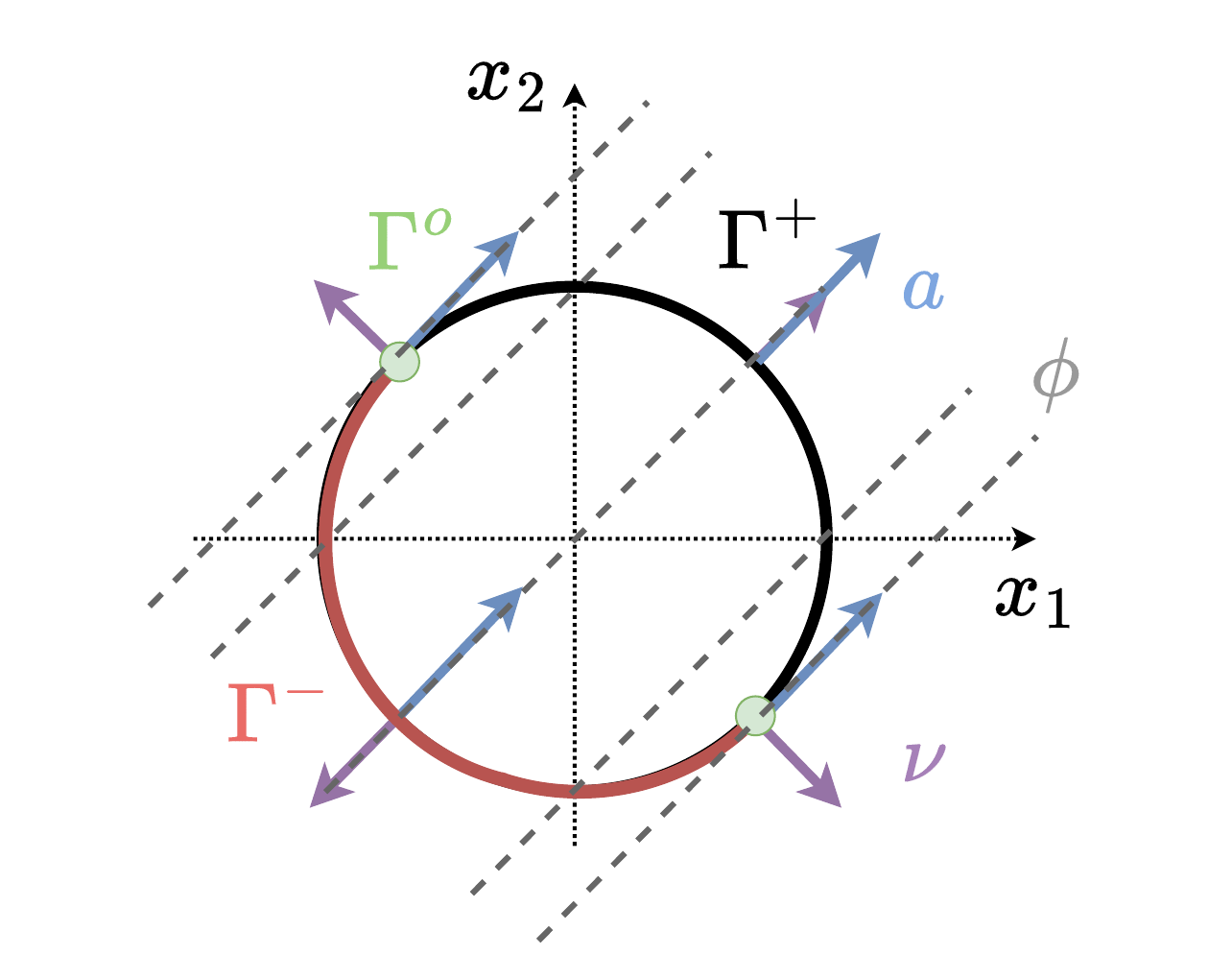}
\end{figure}
For each $x=(x_1,x_2)\in \bar{\Omega}\setminus \Gamma^o$, the entry time of $\phi(\cdot;x)$ is 
\begin{align*}
\tau^+(x) &= \inf \{s\geq 0 \ : \ (s+x_1,s+x_2)\in \Gamma^+\} \\
&= \inf \{s\geq 0 \ : \ (s+x_1)^2+(s+x_2)^2=1 \\ &\qquad \qquad \text{and} \ s+x_1>-(s+x_2)\}.
\end{align*}
The polynomial equation $(s+x_1)^2+(s+x_2)^2=1$ admits two roots 
\begin{align*}
s_{\pm} = \frac{-(x_1+x_2)\pm \sqrt{(x_1+x_2)^2+2(1-x_1^2-x_2^2)}}{2}.
\end{align*}
Additionally, from the condition $s+x_1>-(s+x_2)$, we must have 
\begin{align*}
s > \frac{-(x_1+x_2)}{2}.
\end{align*}
We therefore select the nonnegative root $s_+$, and obtain 
\begin{align*}
\tau^+(x) =  \frac{-(x_1+x_2)+\sqrt{(x_1+x_2)^2+2(1-x_1^2-x_2^2)}}{2}. 
\end{align*}
Similarly, we have 
\begin{align*}
\tau^-(x) = \frac{-(x_1+x_2)- \sqrt{(x_1+x_2)^2+2(1-x_1^2-x_2^2)}}{2}.
\end{align*}
Hence, the transit time of $\phi(\cdot;x)$ is given by 
\begin{align*}
T(x) &= \tau^+(x)-\tau^-(x) \\
&= \sqrt{(x_1+x_2)^2+2(1-x_1^2-x_2^2)} \\
&= \sqrt{2-(x_1-x_2)^2}.
\end{align*}
As a result, Assumption \ref{ass4} holds with 
\begin{align*}
T_{\max} = T(0) = \sqrt{2}. 
\end{align*}
We consider on $\Omega$ the PDE 
\begin{align}
\frac{\partial v}{\partial t} = \frac{\partial v}{\partial x_1} + \frac{\partial v}{\partial x_2} + \gamma \exp^{b(x_1+x_2)} v(\rho(x)), \label{pde_eg} 
\end{align}
where $\gamma,b>0$ are constants. Note that Assumption \ref{ass5} holds. Here, 
\begin{align*}
\rho(x) =&~ \phi(\tau^-(x);x) \\
=&~ (\tau^-(x),\tau^-(x))+x \\
=&~ \bigg(\frac{x_1-x_2- \sqrt{(x_1+x_2)^2+2(1-x_1^2-x_2^2)}}{2},\\
&~\qquad \frac{x_2-x_1- \sqrt{(x_1+x_2)^2+2(1-x_1^2-x_2^2)}}{2}\bigg).
\end{align*}
To achieve stabilization in finite time $T_{\max}=\sqrt{2}$, we implement the control law in \eqref{control_law_x} for all $z = (z_1,z_2)\in \Gamma^+$. The integration path is 
\begin{align*}
&C(z) = \bigg\{(s,s)+z \ : \\
&~\frac{-(z_1+z_2)- \sqrt{(z_1+z_2)^2+2(1-z_1^2-z_2^2)}}{2} \leq s \leq 0\bigg\}.
\end{align*}
It remains to calculate the kernel $\mathcal{K}$ defined in \eqref{gain_K}. For each $(z,y)\in \Gamma^+\times \bar{\Omega}$ such that $y\in C(z)$, we have 
\begin{align*}
\mathcal{K}(z,y) = \frac{\sqrt{2}}{2}K(\sigma(y);\rho(z)),
\end{align*}
where $\sigma(y) = -\tau^-(y)$, and 
\begin{align*}
K(\sigma(y);\rho(z)) =&~ \frac{1}{T(\rho(z))}\,
    k\left(1,\frac{\sigma(y)}{T(\rho(z))};\rho(z)\right).
\end{align*}
For each $\rho = (\rho_1,\rho_2)\in \Gamma^-$, the function $(\bar{\sigma},\bar{y})\mapsto k(\bar{\sigma},\bar{y};\rho)$ verifies, on the interior of $\mathcal{T} = \{(\bar{\sigma},\bar{y})\in [0,1]^2\ : \ 0\leq \bar{y}\leq \bar{\sigma} \leq 1\}$,
\begin{align*}
&\frac{\partial k}{\partial \bar{\sigma}} + \frac{\partial k}{\partial \bar{y}} = 0, \\
&k(\bar{\sigma},0;\rho) = T(\rho)\int_0^{\bar{\sigma}}k(\bar{\sigma},\bar{y};\rho)g(\phi(T(\rho)\bar{y};\rho))d\bar{y} \\
&~\qquad \qquad \quad -T(\rho)g\left(\phi \left(T(\rho)\bar{\sigma};\rho\right) \right),
\end{align*}
where 
\begin{align*}
g\left(\phi \left(T(\rho)\bar{\sigma};\rho\right) \right)=&~ g\left((T(\rho)\bar{\sigma},T(\rho)\bar{\sigma})+\rho\right)\\
=&~ \gamma \exp^{b(\rho_1+\rho_2)}\exp^{2b T(\rho)\bar{\sigma}}.
\end{align*}
The solution to the above equations has the form 
\begin{align*}
k(\bar{\sigma},\bar{y};\rho) = \eta(\bar{\sigma}-\bar{y};\rho),
\end{align*}
where $\eta$ will be determined. One has 
\begin{align*}
&\eta(\bar{\sigma};\rho) \\
&\quad = \gamma T(\rho) \exp^{b(\rho_1+\rho_2)}\int_0^{\bar{\sigma}}\eta(\bar{\sigma}-\bar{y};\rho)\exp^{2bT(\rho)\bar{y}} d \bar{y} \\
&\quad \quad - \gamma T(\rho) \exp^{b(\rho_1+\rho_2)}\exp^{2bT(\rho)\bar{\sigma}}. 
\end{align*}
Applying the Laplace transform in the $\bar{\sigma}$-variable, and denoting by $s\mapsto \hat{\eta}(s;\rho)$ the Laplace transform of $\bar{\sigma}\mapsto \eta(\bar{\sigma};\rho)$, we obtain 
\begin{align*}
\hat{\eta}(s;\rho) =&~ \gamma T(\rho)\exp^{b(\rho_1+\rho_2)}\frac{\hat{\eta}(s;\rho)}{s-2bT(\rho)}  \\
&~-\gamma T(\rho)\exp^{b(\rho_1+\rho_2)}\frac{1}{s-2bT(\rho)}.
\end{align*}
Hence, 
\begin{align*}
\hat{\eta}(s;\rho) = -\frac{\gamma T(\rho)\exp^{b(\rho_1+\rho_2)}}{s-\left(2b+\gamma \exp^{b(\rho_1+\rho_2)} \right)T(\rho)}.
\end{align*}
Applying the inverse Laplace transform, 
\begin{align*}
\eta(\bar{\sigma};\rho) = - \gamma T(\rho)\exp^{b(\rho_1+\rho_2)} \exp^{\left(2b+\gamma \exp^{b(\rho_1+\rho_2)} \right)T(\rho)\bar{\sigma}}.
\end{align*}
Therefore, 
\begin{align*}
&k(\bar{\sigma},\bar{y};\rho) \\
&\quad = - \gamma T(\rho)\exp^{b(\rho_1+\rho_2)} \exp^{\left(2b+\gamma \exp^{b(\rho_1+\rho_2)} \right)T(\rho)(\bar{\sigma}-\bar{y})}.
\end{align*}
As a result, we have an explicit expression for $\mathcal{K}$.

We now evaluate our controller numerically. To accommodate general initial conditions
$v_o$ that are not compatible with the controller $U$, and
inspired by \cite{coron_paper}, we modify the control law as
\begin{align*}
    v(z,t) &= U_{\text{new}}(v(\cdot,t),z)\\
           &= U(v(\cdot,t),z)+\varepsilon(z,t) \quad z\in \Gamma^+,\ t\geq 0,
\end{align*}
where
\begin{align*}
    \frac{\partial \varepsilon}{\partial t} &= -m_1 \operatorname{sign}(\varepsilon)|\varepsilon|^{m_2}
        \qquad m_1>0,\ m_2\in (0,1),\\
    \varepsilon(z,0) &= -U(v_o,z)+v_o(z).
\end{align*}
With this modification, any $v_o$ is compatible with $U_{\text{new}}$.
Moreover, since $\varepsilon$ converges to zero in finite time $t_o>0$ (which can be tuned
freely through the parameters $m_1$ and $m_2$), it follows that $v(\cdot,t_o)$ is compatible
with $U$, thereby guaranteeing that $v(x,t)=0$ for all $x\in \bar{\Omega}$ and all
$t\geq T_{\max}+t_o$.

The parameters are set to $\gamma=2$, $b=0.1$, $m_1=5$, and $m_2=0.5$.
The initial condition is chosen as
\begin{align*}
    v_o(x) = 1-x_1^2-x_2^2, \quad x\in \bar{\Omega},
\end{align*}
yielding $T_{\max}+t_o \approx 2.5$. The kernel $\mathcal{K}$ is depicted in
Figure~\ref{fig_kernel}.

\begin{figure}[ht]
    \centering
    \includegraphics[width=\linewidth]{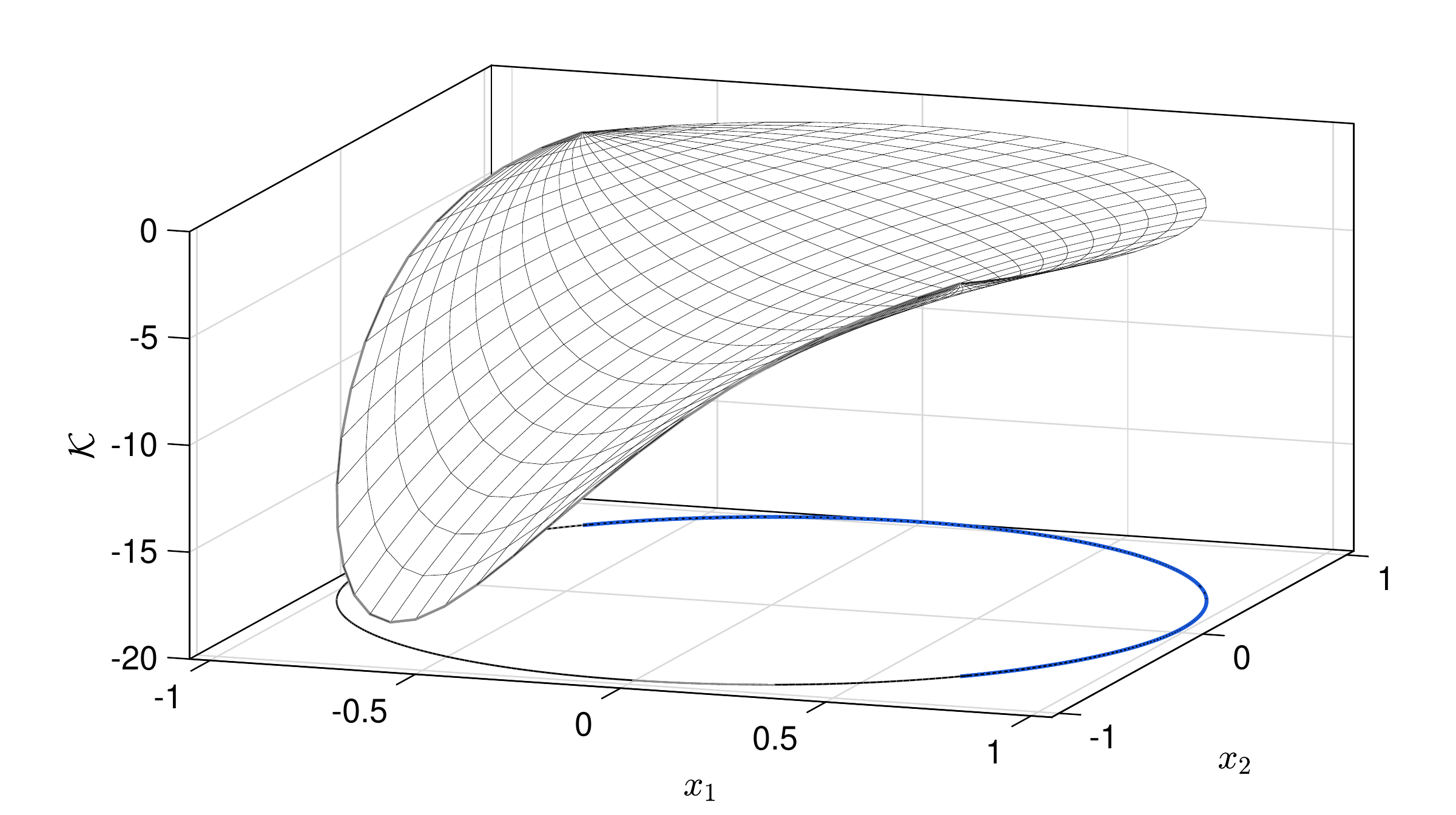}
    \caption{The kernel $\mathcal{K}$ for all $(z,y)\in \Gamma^+\times \bar{\Omega}$ such
             that $y\in C(z)$. The blue curve represents the set $\Gamma^+$.}
    \label{fig_kernel}
\end{figure}

The open-loop and closed-loop responses are then simulated and shown in
Figures~\ref{fig_open} and~\ref{fig_closed}, respectively. As expected, the state grows in open loop, whereas it converges to zero at the finite time $t=2.5$ in closed loop.

\begin{figure}
    \centering
    \includegraphics[width=\linewidth]{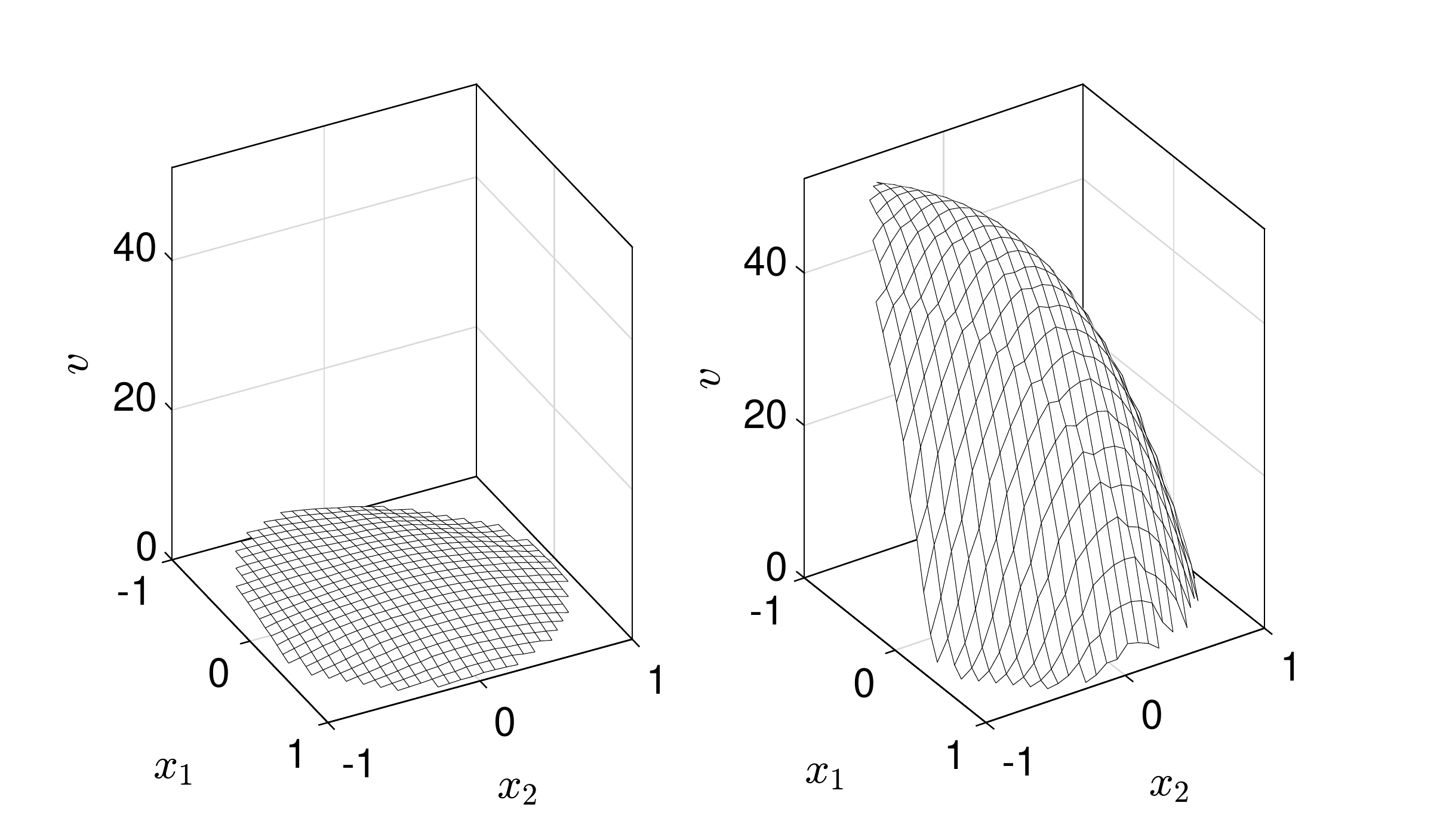}
    \caption{Open-loop response at $t=1$ (left) and $t=2.5$ (right).}
    \label{fig_open}
\end{figure}

\begin{figure}
    \centering
    \includegraphics[width=\linewidth]{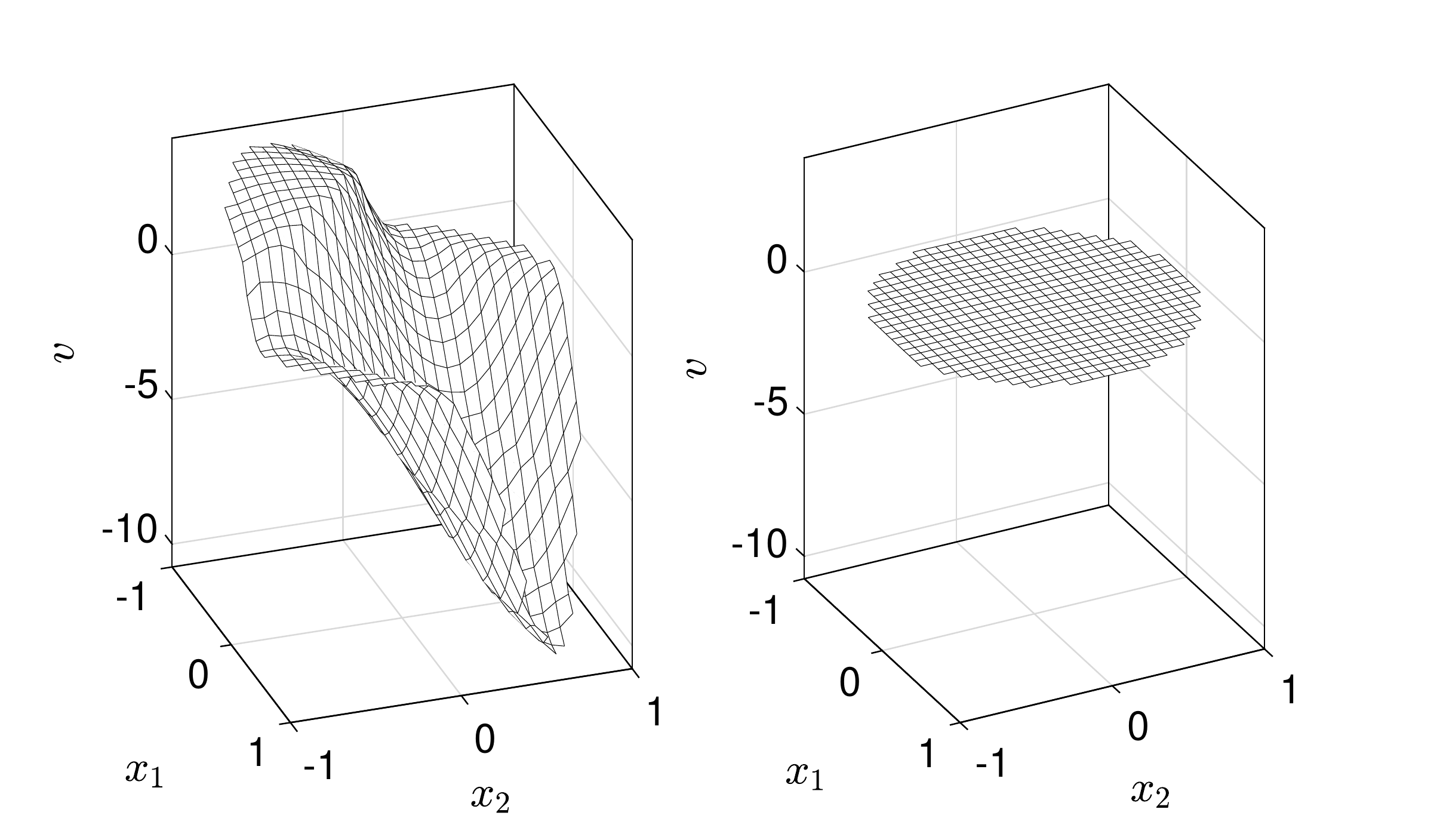}
    \caption{Closed-loop response at $t=1$ (left) and $t=2.5$ (right).}
    \label{fig_closed}
\end{figure}

\section{Conclusion}
This paper extended the backstepping approach of \cite{krstic_H} to first-order hyperbolic equations defined on $n$-dimensional
domains, for any $n \geq 1$, by decomposing the multidimensional equation into
a continuum of one-dimensional equations along characteristic curves and
applying the approach of \cite{krstic_H} to each equation. In a forthcoming work, we reveal how the proposed framework can be extended to interconnections of multidimensional first-order hyperbolic equations with colinear velocities.

It is worth investigating if other control approaches developed in the one-dimensional setting might also extend to the multidimensional case through our reduction. For example, Lyapunov-based techniques constitute one of the most important approaches for the stability analysis of one-dimensional hyperbolic systems \cite{lyapu_1,lyapu_2,lyapu_3,lyapu_4,lyapu_5,lyapu_6,lyapu_7,lyapu_PI}. In the multidimensional setting, several important extensions have been proposed in \cite{lyapu_8,lyapu_9,lyapu_10,lyapu_11} for establishing stability through Lyapunov functionals defined over the entire spatial domain. In light of the reduction established in Theorem~\ref{thm1}, one may attempt to construct a Lyapunov functional for each one-dimensional equation in the resulting continuum and then relate the corresponding stability properties back to those of the original multidimensional equation. Such an approach could, for example, yield bounds on the $\mathcal{C}^0$ norm of the state in dimensions $\geq 2$ without resorting to $H^s$ Lyapunov functionals with $s \geq 2$, relying instead on a continuum of $H^1$ Lyapunov functionals along each leaf. When considering a scalar multidimensional equation, or systems of equations with colinear transport velocities, only one foliation of the domain by characteristic curves needs to be accounted for. In more general settings, however, the Lyapunov analysis would have to accommodate the additional coupling induced by multiple, non-aligned foliations of the domain associated with distinct characteristic directions. Our reduction may also make it possible to extend spectral approaches developed for one-dimensional hyperbolic equations \cite{hastir,dus,camil_2} to the multidimensional setting. This approach would avoid studying an eigenvalue problem directly on the full multidimensional domain, requiring instead only the analysis of the family of one-dimensional equations arising in the continuum reduction. The spectral properties of the original multidimensional system could then be recovered from the corresponding properties of these one-dimensional equations/systems.

\section*{Acknowledgments}
The author is thankful to Prof. M. Krstic for insightful comments on an earlier version of this manuscript, and for pointing out the connection of Assumption~\ref{ass4} with the theory of geometric optics, which motivated the terminology \textit{non-trapping characteristics}.

\par\noindent 
\parbox[t]{\linewidth}{
\noindent\parpic{\includegraphics[height=1.5in,width=1in,clip,keepaspectratio]{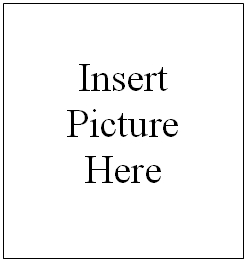}}
\noindent {\bf Mohamed Camil Belhadjoudja}\ was born in Algiers, Algeria, on November 28, 1998. He received the Engineering degree and the M.Sc. degree in automatic control from the National Polytechnic School, Algeria, in 2021, and a second M.Sc. degree in automatic control from Grenoble Alpes University, France, in 2022. He received the Ph.D. degree in automatic control from Grenoble Alpes University in 2025. Since 2026, he has been a Postdoctoral Fellow at the University of Waterloo, Canada. His research interests lie in control and estimation of infinite-dimensional systems.}

\end{document}